\documentclass[10pt,twocolumn,twoside]{IEEEtran}
\usepackage{generic}
\usepackage{cite}
\usepackage{amsmath, amssymb, amsfonts}
\usepackage{amsthm}
\usepackage{algorithmic}
\usepackage{graphicx}
\usepackage{textcomp}
\usepackage{hyperref}
\newtheorem{theorem}{Theorem}
\newtheorem{lemma}{Lemma}

\begin{document}
\title{Consensus in Multi-Agent Systems with Uniform and Nonuniform Communication Delays}%: Theory and Experimental Validation on a Swarm of QBOT3 Robots}
%Nonuniform Delay Consensus in Multi-Agent Systems: Convergence Analysis and QBOT3 Robot Experiments
%Distributed Consensus with Nonuniform Delays: Stability, Convergence, and QBOT3 Robot Implementation
%Consensus Algorithms for Multi-Agent Systems with Heterogeneous Delays: Theory and QBOT3 Swarm Validation
\author{Shokoufeh Naderi, Maude J. Blondin, \IEEEmembership{Member, IEEE}, Sébastien Roy
\thanks{}
\thanks{All authors are with the Department of Electrical Engineering and Computer Engineering, Université de Sherbrooke, Sherbrooke, QC, Canada.}
\thanks{Contacts for the authors: Shokoufeh.Naderi@USherbrooke.ca, Maude.Blondin2@USherbrooke.ca, Sebastien.Roy13@USherbrooke.ca.}
%\thanks{This work was supported by the Natural Sciences and Engineering Research Council of Canada (NSERC) - Discovery Grant.}}
\thanks{This work has been submitted to the IEEE for possible publication. Copyright may be transferred without notice, after which this version may no longer be accessible.}}

\maketitle

\begin{abstract}
This paper analyzes consensus in multi-agent systems under uniform and nonuniform communication delays, a key challenge in distributed coordination with applications to robotic swarms. It investigates the convergence of a consensus algorithm accounting for delays across communication links in a connected, undirected graph. Novel convergence results are derived using Rouché's theorem and Lyapunov-based stability analysis. The system is shown to reach consensus at a steady-state value given by a weighted average determined by the delay distribution, with stability ensured under explicit parameter bounds. Both uniform and nonuniform delay scenarios are analyzed, and the corresponding convergence values are explicitly derived. The theoretical results are validated through simulations, which explore the impact of delay heterogeneity on consensus outcomes. Furthermore, the algorithm is implemented and experimentally tested on a swarm of QBOT3 ground robots to solve the rendezvous problem, demonstrating the agents’ ability to converge to a common location despite realistic communication constraints, thus confirming the algorithm’s robustness and practical applicability. The results provide guidelines for designing consensus protocols that tolerate communication delays, offer insights into the relationship between network delays and coordination performance, and demonstrate their applicability to distributed robotic systems.
\end{abstract}


\section{Introduction}
\label{sec:introduction}

Multi-agent systems have become a powerful approach for solving complex optimization and consensus problems in a distributed manner. In recent years, researchers have extended distributed optimization methods to handle multi-objective problems, in which agents seek to optimize multiple objective functions simultaneously. One such example is the algorithm proposed by Blondin and Hale \cite{blondin2020algorithm}, which enables a network of agents to collaboratively explore Pareto-optimal solutions by exchanging information weighted by each agent's prioritization of others' objectives. That algorithm was proven to converge under ideal, delay-free communication. In practical networks, however, communication delays are unavoidable due to transmission and processing latencies. These delays can degrade or even prevent the convergence of distributed algorithms if not properly accounted for. This paper extends the algorithm of Blondin and Hale by analyzing its convergence properties under constant communication delays. Two delay scenarios are considered: uniform delays, where all communication links have the same constant delay, and nonuniform delays, where different links may have different (but fixed) delays. The focus is on establishing rigorous convergence guarantees under these delays and quantifying how delays affect the algorithm's performance and final outcomes.

The effect of communication delays on consensus and distributed optimization has been extensively studied in the control and multi-agent systems (MASs) literature. Early works showed that consensus (agreement to a common value) can still be achieved in the presence of bounded time delays, though under more restrictive conditions. For instance, Olfati-Saber and Murray showed in \cite{olfati2004consensus} that for continuous-time consensus, if the interaction topology remains connected over time, the agents can reach agreement despite uniform communication delays. Subsequent studies in discrete-time systems provided more quantitative conditions. %In particular, researchers derived delay-dependent stability criteria---often in the form of upper bounds on the allowable delay---to guarantee convergence. 
Bliman and Ferrari-Trecate \cite{bliman2008average} obtained exact delay bounds for average consensus on fixed networks, showing how large a single delay, or multiple constant delays, can be before consensus on the average of the initial states is lost. Similarly, Wang et al. derived an upper bound on delay tolerance for consensus of identical agents \cite{wang2013consensus}. These works, however, generally assumed a uniform delay on all communication links for analytical tractability.

Researchers have gradually relaxed the uniform-delay assumption, recognizing that equal delays on all links are an unrealistic simplification. For instance, Rehák and Lynnyk \cite{rehak2020consensus} showed that for consensus of identical agents with heterogeneous input delays, exact synchronization cannot, in general, be achieved; instead, a bounded steady-state disagreement remains. Similarly, Lin and Jia \cite{lin2009consensus} analyzed second-order consensus protocols in discrete-time systems with nonuniform communication delays and dynamically changing topologies, deriving convergence conditions using properties of nonnegative matrices.

The increased complexity of nonuniform delays often necessitates more advanced analytical tools, such as Lyapunov methods or spectral analysis, to establish stability guarantees. To analyze the convergence of the multi-objective algorithm under delays, this paper employs two complementary analytical techniques: Rouché's theorem from complex analysis and Lyapunov-based stability analysis. In the uniform-delay case, an application of Rouché's theorem to the multi-agent domain is presented. Rouché's theorem has been used in stability analysis of time-delay differential equations \cite{breda2005pseudospectral}, but to our knowledge, it has not been applied to prove convergence of distributed optimization or consensus algorithms. By expressing the algorithm's delayed-update dynamics as a characteristic equation in the complex domain, Rouché's theorem is used to show that all eigenvalues of the update mapping remain inside the unit circle under a certain step-size condition. This approach yields an exact stability condition for the uniform delay case. 

For the more general nonuniform-delay scenario, a Lyapunov-based analysis is used. Lyapunov stability theory is a common tool for handling time-varying delays in networked systems \cite{olfati2004consensus}. A Lyapunov function is designed to account for the delayed information at each communication link, and a Lyapunov-based linear matrix inequality (LMI) condition is used to show that the consensus error decays to zero. 

Another key contribution of this work is determining the exact non-average convergence value of the algorithm under delays. In classical average consensus algorithms without delays, all agents converge to the uniform average of their initial values \cite{bliman2008average, blondin2020algorithm}. Under delays, prior work typically still proves convergence to an average of the initial values \cite{olfati2004consensus, liu2011distributed, olfati2007consensus, spanos2005dynamic, wang2017consensus, li2019consensus} due to the self-delays considered in their update rules. In reality, an agent has essentially zero delay in accessing its own state, while communication delays occur only on links between different agents. Removing this artificial self-delay---while more realistic---breaks the symmetry that preserves the uniform average, making the analysis significantly more challenging and typically resulting in convergence to a weighted (non-uniform) average. %\cite{tian2008consensus, papachristodoulou2010effects}. 
In contrast to prior studies that often retain self-delays, this work explicitly models zero self-delays. This modeling choice preserves convergence while providing a more realistic delay model. The exact (weighted) convergence value is then rigorously derived for both uniform and nonuniform delay cases.

For the uniform-delay case on a regular graph, an explicit and constructive upper bound on the step size is derived that guarantees convergence. While many prior works on consensus and distributed optimization require the step size to be sufficiently small for stability, the resulting conditions are typically implicit — expressed in terms of the spectral radius of the iteration matrix \cite{olshevsky2009convergence, olshevsky2011convergence}, the maximum eigenvalue of the Laplacian \cite{xiao2004fast, nedic2010constrained}, or Lipschitz constants of local objectives \cite{nedic2009distributed, shi2015extra} — and require eigenvalue computations or iterative verification. In contrast, our bound is closed-form, depends only on easily computable graph invariants (maximum and average degree) and the delay, and does not require solving eigenvalues, making it suitable for analysis and implementation in large-scale networks.

Finally, all theoretical results are validated through simulations and physical experiments on a multi-robot testbed. The delayed-information consensus algorithm is implemented on a set of Quanser QBot3 mobile robots that communicate over Wi-Fi with induced constant latency to solve the rendezvous problem. The experimental results confirm robustness to communication delays.

The rest of the paper is organized as follows. Section~\ref{sec:preliminaries} reviews the multi-objective multi-agent optimization algorithm from \cite{blondin2020algorithm} and introduces our delay models and problem formulation. Section~\ref{sec:main_results} presents the convergence analysis for the uniform delay case using Rouché’s theorem, the nonuniform delay case using a Lyapunov-based analysis, and the derivation of explicit convergence values and the step size bound. Section~\ref{sec:results} provides simulation and experimental results that validate the theoretical findings. Finally, Section~\ref{sec:conclusion} concludes the paper and outlines future directions.

\section{Preliminaries and Problem Description} \label{sec:preliminaries}
Consider a MAS comprising $ n $ agents that interact over a connected, undirected graph $ \mathcal{G} = (\mathcal{V}, \mathcal{E}) $, where $ \mathcal{V} = \{1, \ldots, n\} $ represents the set of agents, and $ \mathcal{E} \subseteq \mathcal{V} \times \mathcal{V} $ denotes the set of communication links. An edge $ (i, j) \in \mathcal{E} $ indicates that agents $ i $ and $ j $ are neighbors and can exchange information directly. The graph has no self-loops, i.e., $ (i, i) \notin \mathcal{E} $ for all $ i \in \mathcal{V} $, though each agent has access to its own data. The adjacency matrix $ H(\mathcal{G}) = [h_j^i] \in \mathbb{R}^{n \times n} $ is defined with entries $ h_j^i = 1 $ if $ (i, j) \in \mathcal{E} $, and 0 otherwise. The degree of agent $ i $, denoted $ \deg(i) $, is the number of its neighbors, given by $ \deg(i) = \sum_{j=1}^n h_j^i $. The degree matrix $ \Delta(\mathcal{G}) \in \mathbb{R}^{n \times n} $ is diagonal, with $ \Delta^i_i = \deg(i) $ for $ i = 1, \ldots, n $, and the maximum degree is $ \Delta_{\max} = \max_{i \in [n]} \deg(i) $. The Laplacian matrix of the graph is defined as $ L(\mathcal{G}) = \Delta(\mathcal{G}) - H(\mathcal{G}) $, capturing the network’s connectivity properties. For the remainder of this paper, the graph notation $ \mathcal{G} $ is omitted, and the Laplacian, degree, and adjacency matrices are simply denoted as $ L $, $ \Delta $, and $ H $, respectively. Throughout the paper, for a matrix $ M = [m_j^i] $, the notation $ m_j^i $ (or $ [M]_j^i $) denotes the entry in the $i$-th row and $j$-th column of $M$.

Each agent $ i $ maintains a priority vector $ w^i(k) = [w_1^i(k), w_2^i(k), \ldots, w_n^i(k)] \in \mathbb{R}^{1 \times n} $ at time step $ k $, which evolves according to the update rule
\begin{equation} \label{eq:update_rule}
    w^i(k+1) = w^i(k) + c \sum_{j=1}^n h_j^i (w^j(k - \tau_{ij}) - w^i(k))
\end{equation}
for $ k \geq \tau_{\max} $, with initial conditions $ w^i(k) = w^i(0) \in \mathbb{R}^n $ for $ k = 0, 1, \ldots, \tau_{\max} - 1 $. Here, $ \tau_{\max} = \max_{(i,j) \in \mathcal{E}} \tau_{ij} $ is the maximum communication delay, and $ \tau_{ij} $ is the nonnegative integer delay from agent $ j $ to agent $ i $, satisfying $ \tau_{ij} = \tau_{ji} $ (due to the undirected nature of the graph) and $ \tau_{ii} = 0 $. The priority vectors are used to weigh local objective functions in a multi-objective optimization (MOO) problem, formulated as
\begin{equation} \label{eq:optimization}
    \min_{\mathbf{x} \in \mathbf{A}} \mathbf{f(x)} = \sum_{i=1}^n w_i f_i(\mathbf{x}),
\end{equation}
where $ \mathbf{f(x)} $ is the global objective function, $ f_i: \mathbb{R}^m \to \mathbb{R} $ is the local objective function known only to agent $ i $, and $m$ is the number of decision variables. $ \mathbf{A} \subseteq \mathbb{R}^m $ is the feasible region defined by the problem constraints, and $ \mathbf{x} \in \mathbf{A} $ represents a feasible solution. The weight $ w_i = [w^i]_i $ reflects the priority of $ f_i $. %, with $ \sum_{i=1}^n w_i = 1 $.

The objective is to achieve consensus on the priority vectors, such that $ w^i(k) \to \alpha \in \mathbb{R}^{1 \times n} $ for all $ i $, enabling the agents to collaboratively solve the optimization problem \eqref{eq:optimization} with agreed-upon weights. However, the presence of nonuniform communication delays $ \tau_{ij} $, which vary across edges, poses a significant challenge to the consensus process, as delayed information exchange disrupts coordination among agents. Traditional consensus algorithms often assume uniform delays or instantaneous communication, which limits their effectiveness in practical scenarios like robotic swarms, where delays are heterogeneous due to network constraints. This paper addresses the problem of achieving consensus in MASs under uniform and nonuniform delays, analyzing the convergence properties and steady-state behavior of the system defined by \eqref{eq:update_rule}, with practical validation through the rendezvous problem in a swarm of QBOT3 ground robots.

\section{Main Results} \label{sec:main_results}
This section presents the core theoretical findings on the consensus behavior of the MAS described in the previous section, focusing on the effects of communication delays. 
The following theorems establish the exact consensus value and its properties. To enhance clarity and simplify the analysis, the most basic case involving a uniform delay of $ \tau = d $ on a regular graph is first presented in Theorem~\ref{thm:delay_consensus_d1}. This result is then extended to handle nonuniform delays in Theorem~\ref{thm:delay_consensus_nonuniform}. Prior to stating the theorems, Lemma~\ref{lem:vector_update_nonuniform_delay} and Lemma~\ref{lem:conservation_nonuniform} provide foundational results that are essential for the subsequent proofs.

%The following assumption is made on the local objective functions:
%\begin{assumption} %\label{assumption1}
%    For all $i\in\{1,...,n\}$, the function $f_i:\mathbb R^n \to \mathbb R$ is continuously differentiable and convex.
%\end{assumption}

\begin{lemma} \label{lem:vector_update_nonuniform_delay}
    Consider the MAS described in the previous section, where each agent $ i $ updates its priority vector $ w^i(k) $ according to \eqref{eq:update_rule}. Define the network state matrix $ W(k) = [w^1(k); w^2(k); \ldots; w^n(k)] \in \mathbb{R}^{n \times n} $. Then, the update rule in matrix form is
    \begin{equation} \label{lemma1_Dynamics}
        W(k+1) = (I - c \Delta) W(k) + c \sum_{m=1}^{\tau_{\max}} H_m W(k - m),
    \end{equation}
    where $ I $ is the $ n \times n $ identity matrix, and $ H_m = [h_j^i \cdot \mathbf{1}_{\{\tau_{ij} = m\}}] $ is the adjacency matrix for edges with delay $ m $, for $ m = 1, \ldots, \tau_{\max} $, with $ \mathbf{1}_{\{\tau_{ij} = m\}} = 1 $ if $ \tau_{ij} = m $, and 0 otherwise.
\end{lemma}

\begin{proof}
    To derive the matrix form of the update rule, the agent-level dynamics are transformed into a network-level representation, utilizing the adjacency structure defined by $ H_m $ and the degree matrix $ \Delta $.

    Begin with the agent-level update rule defined in \eqref{eq:update_rule}. Since $ w^i(k) $ is independent of $ j $,  the second term inside the summation, $ -c \sum_{j=1}^n h_j^i w^i(k) $, simplifies to
    \begin{equation}
        \sum_{j=1}^n h_j^i w^i(k) = w^i(k) \sum_{j=1}^n h_j^i = \deg(i) w^i(k),
    \end{equation}
    where $ \deg(i) = \sum_{j=1}^n h_j^i $ is the degree of node $ i $, quantifying the number of neighbors. $ h_j^i = 1 $ if $ (i, j) \in E $ and 0 otherwise (noting that $ h_i^i = 0 $). In matrix form, this term across all agents becomes $ -c \Delta W(k) $, where $ \Delta $ is the diagonal matrix with $ \Delta_i^i = \deg(i) $, and the $ i $-th row of $ \Delta W(k) $ is $ \deg(i) w^i(k) $.

    Now, consider the delayed term $ c \sum_{j=1}^n h_j^i w^j(k - \tau_{ij}) $, which aggregates contributions from neighbors with varying delays $ \tau_{ij} $. Since $ \tau_{ij} $ ranges from 1 to $ \tau_{\max} $, this sum is partitioned by delay values. For each delay $ m $ from 1 to $ \tau_{\max} $, the contribution from neighbors with $ \tau_{ij} = m $ is $ c \sum_{j: \tau_{ij} = m} h_j^i w^j(k - m) $ \footnote{$ \sum_{j: \tau_{ij} = m} h_j^i w^j(k - m) $ includes only those neighbors $ j $ of $ i $ where the delay $ \tau_{ij} = m $. For example, $ c H_1 W(k - 1) $ represents the influence from neighbors with delay $ \tau_{ij} = 1 $. The $ i $-th row of $ H_1 W(k - 1) $ is
    \begin{equation}
    [H_1 W(k - 1)]^i = \sum_{j=1}^n [H_1]_{j}^i w^j(k - 1) = \sum_{j: \tau_{ij} = 1} h_j^i w^j(k - 1).
    \end{equation}
    This is the weighted sum of the previous state $ w^j(k - 1) $ for all $ j $ connected to $ i $ with a delay of 1.}. The total delayed influence is thus
    \begin{equation}
        c \sum_{j=1}^n h_j^i w^j(k - \tau_{ij}) = c \sum_{m=1}^{\tau_{\max}} \sum_{j: \tau_{ij} = m} h_j^i w^j(k - m).
    \end{equation}
    This double summation groups the neighbor contributions by their respective delay lags. In matrix notation, the inner sum $ \sum_{j: \tau_{ij} = m} h_j^i w^j(k - m) $ corresponds to the $ i $-th row of $ H_m W(k - m) $, where $ H_m = [h^i_j \cdot \mathbf{1}_{\{\tau_{ij} = m\}}] $ is the adjacency matrix isolating edges with delay $ m $, and $ \mathbf{1}_{\{\tau_{ij} = m\}} $ is an indicator function ensuring only the appropriate $ j $ contribute. The outer sum over $ m $ from 1 to $ \tau_{\max} $ aggregates all delayed terms, yielding $ c \sum_{m=1}^{\tau_{\max}} H_m W(k - m) $, where each $ H_m W(k - m) $ represents the influence from neighbors at lag $ m $, weighted by the network structure at that delay.

    Combining these terms, the agent-level update becomes
    \begin{equation}
        w^i(k+1) = w^i(k) + c \sum_{m=1}^{\tau_{\max}} [H_m W(k - m)]^i - c [\Delta W(k)]^i.
    \end{equation}
    Aggregating across all agents, the network state matrix $ W(k) = [w^1(k); w^2(k); \ldots; w^n(k)] $, where $ W(k)^i_j = [w^i(k)]_j $ is the $ j $-th component of agent $ i $'s priority vector, evolves according to
    \begin{equation}
        W(k+1) = W(k) + c \sum_{m=1}^{\tau_{\max}} H_m W(k - m) - c \Delta W(k).
    \end{equation}
    This can be equivalently expressed as
    \begin{equation} \label{eq:matrix_form}
        W(k+1) = (I - c \Delta) W(k) + c \sum_{m=1}^{\tau_{\max}} H_m W(k - m).
    \end{equation}
    
    %In the special case of uniform delays with $ \tau_{ij} = \tau $ for all $ (i, j) \in E $, the update rule simplifies. If $ \tau = 1 $, then $ \tau_{\max} = 1 $, and the indicator function $ \mathbf{1}_{\{\tau_{ij} = m\}} $ is 1 only when $ m = 1 $, making $ H_m = 0 $ for all $ m \neq 1 $, and $ H_1 = H $, the full adjacency matrix. The summation in \eqref{eq:matrix_form} reduces to $ c \sum_{m=1}^{\tau_{\max}} H_m W(k - m) = c H W(k - 1) $, and the update rule becomes:
    %\begin{equation} \label{eqlemma1}
    %    W(k+1) = (I - c \Delta) W(k) + c H W(k - 1),
    %\end{equation}
    %which corresponds to the matrix form for a uniform delay of 1, consistent with the dynamics of a system where all communication delays are identical and equal to one time step.
\end{proof}

\begin{lemma}
\label{lem:conservation_nonuniform}
    For the system described in \eqref{lemma1_Dynamics}, define the augmented state vector $ z^i(k) \in \mathbb{R}^{n (\tau_{\max} + 1) \times 1} $ as
    \[
        z^i(k) = [w^i(k)^T, w^i(k-1)^T, \ldots, w^i(k - \tau_{\max})^T]^T.
    \]
    The system dynamics are governed by
    \[
        z^i(k+1) = A z^i(k),
    \]
    where the state transition matrix $A \in \mathbb{R}^{n (\tau_{\max} + 1) \times n (\tau_{\max} + 1)}$ is
    \[
        A = \begin{bmatrix}
            I - c \Delta & c H_1 & c H_2 & \cdots & c H_{\tau_{\max}} \\
            I & 0 & 0 & \cdots & 0 \\
            0 & I & 0 & \cdots & 0 \\
            \vdots & \vdots & \ddots & \ddots & \vdots \\
            0 & 0 & \cdots & I & 0
        \end{bmatrix},
    \]
    where $H = \sum_{l=1}^{\tau_{\max}} H_l$ the adjacency matrix, and $H_l$ the adjacency matrix for edges with delay $l$ (i.e., $(H_l)_j^i = 1$ if $\tau_{ij} = l$, and 0 otherwise). Define the vector $ V \in \mathbb{R}^{n (\tau_{\max} + 1)} $ as
    \begin{equation} \label{eq:Eigenvector}
        V = \left[ \beta \mathbf{1}^T, \beta c \mathbf{1}^T \sum_{l=1}^{\tau_{\max}} H_l, \beta c \mathbf{1}^T \sum_{l=2}^{\tau_{\max}} H_l, \ldots, \beta c \mathbf{1}^T H_{\tau_{\max}} \right]^T,
    \end{equation}
    where $\beta \neq 0$ is a scalar, and $\mathbf{1} = [1, \ldots, 1]^T \in \mathbb{R}^n$ is the column vector of ones. In the rest of the paper, as is the case here, \(\mathbf{1}\) denotes a column vector of ones with dimension \(n\) unless otherwise specified. Then, $V$ is a left eigenvector of $A$ with eigenvalue 1 (i.e., $V^T A = V^T$), and the quantity
    \[
        V^T z^i(k) = \beta \mathbf{1}^T w^i(k) + \beta c \sum_{m=1}^{\tau_{\max}} \left( \mathbf{1}^T \sum_{l=m}^{\tau_{\max}} H_l \right) w^i(k-m)
    \]
    is conserved, implying that $V^T z^i(k) = V^T z^i(0)$ for all $k \geq 0$.
\end{lemma}

\begin{proof}
    %The goal is to show that the quantity $V^T z^i(k)$ remains constant over time for each agent $i$. This indicates that a weighted sum of the agent’s states (current and delayed) is preserved under the system dynamics. %The augmented state $z^i(k)$ stacks the current state $w^i(k)$ and its delayed versions up to $\tau_{\max}$, allowing us to model nonuniform delays in a unified state-space form. The matrix $A$ captures how the current state updates via a consensus step (using the degree matrix $\Delta$ and adjacency matrices $H_l$) while shifting delayed states forward.
    To establish conservation, it is first verified that $V$ is a left eigenvector of $A$ with eigenvalue 1. This means $V^T A = V^T$, implying that applying $A$ to $z^i(k)$ does not change the projection along $V^T$. Expressing $V^T$ in block form, where each block is a row vector in $\mathbb{R}^{1 \times n}$, gives
    \[
        V^T = \left[ \beta \mathbf{1}^T, \beta c \mathbf{1}^T \sum_{l=1}^{\tau_{\max}} H_l, \beta c \mathbf{1}^T \sum_{l=2}^{\tau_{\max}} H_l, \ldots, \beta c \mathbf{1}^T H_{\tau_{\max}} \right].
    \]
    The matrix $A$ has $\tau_{\max} + 1$ block rows and columns, where each block is of size $n \times n$. We can then determine $V^T A$ by multiplying each block of $V^T$ by the corresponding column blocks of $A$, as follows:
    \begin{align*}
        V^T A &= [ \beta \mathbf{1}^T, \beta c \mathbf{1}^T \sum_{l=1}^{\tau_{\max}} H_l, \ldots, \beta c \mathbf{1}^T \sum_{l=\tau_{\max}-1}^{\tau_{\max}} H_l, \\
        &\beta c \mathbf{1}^T H_{\tau_{\max}} ]
        \begin{bmatrix}
            I - c \Delta & c H_1 & c H_2 & \cdots & c H_{\tau_{\max}} \\
            I & 0 & 0 & \cdots & 0 \\
            0 & I & 0 & \cdots & 0 \\
            \vdots & \vdots & \ddots & \ddots & \vdots \\
            0 & 0 & \cdots & I & 0
        \end{bmatrix}.
    \end{align*}
    
    \begin{itemize}
        \item \textbf{First component (block 1)}: This corresponds to the first block column of $A$, $[I - c \Delta, I, 0, \ldots, 0]^T$:
        \begin{align*}
            &\beta \mathbf{1}^T (I - c \Delta) + \beta c \left( \mathbf{1}^T \sum_{l=1}^{\tau_{\max}} H_l \right) I \\
            %&+ \beta c \left( \mathbf{1}^T \sum_{l=2}^{\tau_{\max}} H_l \right) \cdot 0 + \cdots + \beta c \left( \mathbf{1}^T H_{\tau_{\max}} \right) \cdot 0 \\
            & = \beta \mathbf{1}^T - \beta c \mathbf{1}^T \Delta + \beta c \mathbf{1}^T \sum_{l=1}^{\tau_{\max}} H_l.
        \end{align*}
        Since $ H = \sum_{l=1}^{\tau_{\max}} H_l$ (the degree matrix counts all edges incident to each node, aggregated across all delays), we have $\mathbf{1}^T \Delta = \mathbf{1}^T H = \mathbf{1}^T \sum_{l=1}^{\tau_{\max}} H_l$. Thus, $ -\beta c \mathbf{1}^T \Delta + \beta c \mathbf{1}^T \sum_{l=1}^{\tau_{\max}} H_l = 0$. So, the first component is $
        \beta \mathbf{1}^T $, matching the first block of $V^T$.
    
        \item \textbf{Second component (block 2)}: This corresponds to the second block column of $A$, $[c H_1, 0, I, 0, \ldots, 0]^T$:
        \begin{align*}
            &\beta \mathbf{1}^T (c H_1) %+ \beta c \left( \mathbf{1}^T \sum_{l=1}^{\tau_{\max}} H_l \right) \cdot 0 
            + \beta c \left( \mathbf{1}^T \sum_{l=2}^{\tau_{\max}} H_l \right) I \\
            %&+ \beta c \left( \mathbf{1}^T \sum_{l=3}^{\tau_{\max}} H_l \right) \cdot 0 + \cdots + \beta c \left( \mathbf{1}^T H_{\tau_{\max}} \right) \cdot 0 \\
            %= \beta c \mathbf{1}^T H_1 + \beta c \mathbf{1}^T \sum_{l=2}^{\tau_{\max}} H_l 
            &= \beta c \mathbf{1}^T \left( H_1 + \sum_{l=2}^{\tau_{\max}} H_l \right) 
            = \beta c \mathbf{1}^T \sum_{l=1}^{\tau_{\max}} H_l,
        \end{align*}
        which matches the second block of $V^T$.
    
        \item \textbf{Third component (block 3)}: This corresponds to the third block column, $[c H_2, 0, 0, I, 0, \ldots, 0]^T$:
        \begin{align*}
            &\beta \mathbf{1}^T (c H_2) %+ \beta c \left( \mathbf{1}^T \sum_{l=1}^{\tau_{\max}} H_l \right) \cdot 0 + \beta c \left( \mathbf{1}^T \sum_{l=2}^{\tau_{\max}} H_l \right) \cdot 0 \\
            + \beta c \left( \mathbf{1}^T \sum_{l=3}^{\tau_{\max}} H_l \right) I \\%+ \cdots + \beta c \left( \mathbf{1}^T H_{\tau_{\max}} \right) \cdot 0 \\
            &= \beta c \mathbf{1}^T H_2 
            + \beta c \mathbf{1}^T \sum_{l=3}^{\tau_{\max}} H_l  \\
            &= \beta c \mathbf{1}^T \left( H_2 + \sum_{l=3}^{\tau_{\max}} H_l \right) = \beta c \mathbf{1}^T \sum_{l=2}^{\tau_{\max}} H_l,
        \end{align*}
        matching the third block of $V^T$.
    
        \item \textbf{General $m$-th component ($m = 2, \ldots, \tau_{\max}$)}: The $m$-th block column of $A$ has $c H_{m-1}$ in the first block, $I$ in the $(m-1)$-th block (if $m-1 \geq 2$), and zeros elsewhere:
        \begin{align*}
            &\beta \mathbf{1}^T (c H_{m-1}) + \beta c \left( \mathbf{1}^T \sum_{l=m}^{\tau_{\max}} H_l \right) I \\
            &= \beta c \mathbf{1}^T H_{m-1} 
            + \beta c \mathbf{1}^T \sum_{l=m}^{\tau_{\max}} H_l = \beta c \mathbf{1}^T \sum_{l=m-1}^{\tau_{\max}} H_l,
        \end{align*}
        which matches the $m$-th block of $V^T$.
    
        \item \textbf{Final component (block $\tau_{\max} + 1$)}: The last block column of $A$ is $[c H_{\tau_{\max}}, 0, \ldots, 0]^T$, so
        \[
            \beta \mathbf{1}^T (c H_{\tau_{\max}}) + 0 + \cdots + 0 = \beta c \mathbf{1}^T H_{\tau_{\max}},
        \]
        matching the last block of $V^T$.
    \end{itemize}
    
    Thus, $V^T A = V^T$, confirming that $V$ is a left eigenvector of $A$ with eigenvalue 1. %This property is crucial because it suggests that the projection of the state along $V$ may be invariant under the dynamics.
    
    Given the system dynamics $z^i(k+1) = A z^i(k)$, apply $V^T$ to both sides
    \[
        V^T z^i(k+1) = V^T (A z^i(k)) = (V^T A) z^i(k).
    \]
    Since $V^T A = V^T$, this becomes
    \[
        V^T z^i(k+1) = V^T z^i(k).
    \]
    This recurrence holds for all $k$. Starting at $k = 0$, we have
    \[
        V^T z^i(1) = V^T z^i(0), \; V^T z^i(2) = V^T z^i(1) = V^T z^i(0), \; \ldots,
    \]
    so by induction, $V^T z^i(k) = V^T z^i(0)$ for all $k \geq 0$. The quantity $V^T z^i(k)$ is therefore conserved.
    
    The conserved quantity is
    \begin{align*}
        V^T z^i(k) &= \beta \mathbf{1}^T w^i(k) + \beta c \left( \mathbf{1}^T \sum_{l=1}^{\tau_{\max}} H_l \right) w^i(k-1) \\
        &+ \beta c \left( \mathbf{1}^T \sum_{l=2}^{\tau_{\max}} H_l \right) w^i(k-2) + \cdots \\
        &+ \beta c \left( \mathbf{1}^T H_{\tau_{\max}} \right) w^i(k - \tau_{\max}),
    \end{align*}
    which can be written as
    \[
        V^T z^i(k) = \beta \mathbf{1}^T w^i(k) + \beta c \sum_{m=1}^{\tau_{\max}} \left( \mathbf{1}^T \sum_{l=m}^{\tau_{\max}} H_l \right) w^i(k-m).
    \]
    This expression weighs the current state and delayed states by factors that depend on the graph structure and delays, reflecting the influence of delayed interactions in the consensus process.
\end{proof}

In the next two subsections, the convergence proofs and consensus values for MASs with uniform delays and non-uniform delays are presented. The uniform delay proof (Lemma~\ref{lem:convergence_uniform_delay}) assumes a regular graph, where all agents have the same degree, enabling spectral analysis via Rouché’s theorem to derive a precise step size condition for stability. This assumption simplifies the eigenvalue analysis by ensuring the degree matrix \( \Delta \) commutes with the Laplacian’s eigenvectors, making \( V^T \Delta V \) diagonal, which is particularly applicable to structured networks such as rings or grids.

The non-uniform delay proof (Lemma~\ref{lem:Lyapunov_Stability_with_LMI}) applies to general graph topologies without requiring regularity and employs Lyapunov theory to establish asymptotic stability. While this approach covers uniform delays as a special case, it does not use the structural properties of regular graphs to provide an explicit step size condition. Presenting both proofs offers complementary insights: the uniform-delay proof provides a precise design parameter for regular graphs, while the non-uniform-delay proof ensures robustness for arbitrary graphs and heterogeneous delays, addressing a broader range of practical scenarios in distributed systems.

\subsection{Uniform delays}
\begin{lemma} \label{lem:convergence_uniform_delay}
    Consider the MAS in Lemma~\ref{lem:vector_update_nonuniform_delay} with a uniform delay $ \tau_{ij} = d $ for all $ (i, j) \in \mathcal{E} $ and $ k \geq 1 $. The system evolves according to the dynamics
    \begin{equation}
        W(k+1) = (I - c \Delta) W(k) + c H W(k - d),
    \end{equation}
    where $ w^i(k) = w^i(0) $ for $ k = -d + 1, \ldots, -1, 0 $, and $ \delta_1 = \frac{1}{n} \sum_{k=1}^n \deg(k) $ is the average degree of the graph. If $ 0 < c < \min\left( \frac{1}{d\delta_1}, \frac{2}{\Delta_{\max}} \right) $, then the system converges to a consensus, i.e., for each component $ j = 1, \ldots, n $ and some $ \alpha_j \in \mathbb{R} $, $ w_j(k) \to \alpha_j \mathbf{1} $ as $ k \to \infty $. Here, $ w_j(k) = [w_j^1(k),\ w_j^2(k),\ \dots,\ w_j^n(k)]^T \in \mathbb{R}^n $ denotes the $j$-th column of the state matrix $W(k)$, so $w_j(k) \to \alpha_j \mathbf{1}$ means that, as $ k \to \infty $, all agents share the same value $\alpha_j$ for the $j$-th objective (i.e., $w_j^1(k) = w_j^2(k) = \cdots = w_j^n(k) = \alpha_j$).
\end{lemma}

\begin{proof}
    See Appendix I.
\end{proof}

\subsubsection{Effect of step size $ c $ and comparison of stability bounds}
The step size $ c $ in the consensus algorithm update rule, given by \eqref{eq:update_rule} plays a critical role in balancing convergence speed and stability. The term $ \sum_{j=1}^n h_j^i (w^j(k - \tau_{ij}) - w^i(k)) $ in \eqref{eq:update_rule} computes a weighted sum of differences between node $ i $’s current state $ w^i(k) $ and the delayed states of its neighbors $ w^j(k - \tau_{ij}) $, where $ h_j^i $ represents the connection weight from node $ j $ to node $ i $. This sum acts as a corrective signal: if $ w^i(k) $ is less than its neighbors’ delayed states, the sum is positive, increasing $ w^i(k+1) $ ; if greater, the sum is negative, decreasing $ w^i(k+1) $. The step size $ c $ directly scales this corrective signal, determining the magnitude of the adjustment applied to $ w^i(k) $. A larger $ c $, makes the update bigger, pushing the node’s state more aggressively toward its neighbors, while a smaller $ c $ makes the update smaller, adjusting the node’s state more gradually. However, because the updates rely on delayed information, a large $ c $ might cause nodes to overshoot their targets. For example, if two nodes are adjusting based on each other’s past states, they could overcorrect and oscillate instead of settling down. In extreme cases, this could prevent convergence altogether.
Thus, $ c $ controls how aggressively node $ i $ responds to discrepancies with its neighbors, balancing the trade-off between rapid alignment and the risk of instability introduced by relying on delayed information.

A larger bound on $ c $ %, such as that permitted by the Rouch\'e theorem, 
accelerates convergence% by reducing the spectral radius of the iteration matrix within the stable range
, thereby decreasing the number of steps required to reach consensus. Conversely, a smaller $ c $ limits the adjustment per step, which comes at the cost of slower convergence, %as the spectral radius remains closer to 1, 
necessitating more iterations. The seminal work by \cite{olfati2007consensus} introduced a bound on $ c $ using the Gershgorin Circle Theorem, yielding $ c < \frac{1}{\Delta_{\max}} $, which is more conservative than the Rouch\'e bound. Adopting the Rouch\'e bound thus allows for larger step sizes, optimizing convergence speed while still guaranteeing stability, offering a practical advantage over the one found through the Gershgorin Circle Theorem, especially in systems where rapid consensus is prioritized.

\begin{theorem} \label{thm:delay_consensus_d1}
    Consider the MAS in Lemma~\ref{lem:convergence_uniform_delay}, where the system converges to a consensus value $ \alpha_j $. The consensus value is given by
    \begin{equation} \label{eqConcensusValueUniform}
        \alpha_j = \frac{\sum_{i=1}^n (1 + c d \deg(i)) w_j^i(0)}{n + 2 c d |E|},
    \end{equation}
    where $|E|$ is the number of edges in $G$. This is a weighted average of the initial conditions with weights $\frac{1 + c d \deg(i)}{n + 2 c d |E|}$, that sum to 1. Then, $\alpha_j \neq \frac{1}{n} \sum_{i=1}^n w_j^i(0)$ (the uniform average) unless $G$ is regular (all degrees equal), and the consensus value satisfies $\min_i w_j^i(0) \leq \alpha_j \leq \max_i w_j^i(0)$.
\end{theorem}

\begin{proof}
    % Modeling the system
    The proof is performed by modeling the delayed system in an augmented state space, using the dynamics from Lemma~\ref{lem:vector_update_nonuniform_delay}. The system’s update rule, with uniform delay $\tau_{ij} = d$ for all $(i,j) \in \mathcal{E}$, is
    \begin{equation*}
        W(k+1) = (I - c \Delta) W(k) + c H W(k - d).
    \end{equation*}
    %, with $W(k)_{ij} = w_j^i(k)$ for $k = -d, \ldots, -1$.

    % Steady-state condition
    At convergence to consensus, $w_j^i(k) \to \alpha_j$ for all $i$ and each $j$, implying $w^i(k) \to \alpha^T = [\alpha_1, \ldots, \alpha_n]$, and thus $W(k) \to W^* = \mathbf{1} \alpha^T$, where $ W^* = [\alpha^T; \alpha^T; \ldots; \alpha^T] $ has all rows equal to $ \alpha^T $. At steady state, $W(k - d) \to W^*$, so 
    \begin{equation} \label{eq:w_star}
        W^* = (I - c \Delta) W^* + c H W^*.
    \end{equation} 
    Substituting $W^* = \mathbf{1} \alpha^T$ into \eqref{eq:w_star} and rearranging terms, we get
    \begin{equation*}
        (I - c \Delta) W^* = \mathbf{1} \alpha^T - c \Delta \mathbf{1} \alpha^T,
    \end{equation*}
    \begin{equation*}
        \quad H W^* = H \mathbf{1} \alpha^T,
    \end{equation*}
    where $\Delta \mathbf{1} = [\deg(1), \deg(2), \ldots, \deg(n)]^T$ and $H \mathbf{1} = \Delta \mathbf{1}$ (since $[H \mathbf{1}]^i = \sum_{j=1}^n h_j^i = \deg(i)$ for an undirected graph). Thus
    \begin{equation*}
        W^* = \mathbf{1} \alpha^T - c \Delta \mathbf{1} \alpha^T + c H \mathbf{1} \alpha^T = \mathbf{1} \alpha^T,
    \end{equation*}
    confirming the steady-state condition.

    % Augmented state space
    To find $\alpha_j$, define the augmented state $z^i(k) = [w^i(k), w^i(k-1), \ldots, w^i(k-d)]^T \in \mathbb{R}^{(d+1)n}$ for agent $i$, that stacks the current and previous priority vectors as columns, with dynamics
    \begin{equation*}
        z^i(k+1) = \begin{bmatrix} (I - c \Delta) w^i(k) + c H w^i(k - d) \\ w^i(k) \\ \vdots \\ w^i(k-d+1) \end{bmatrix} = A z^i(k),
    \end{equation*}
    \begin{equation} \label{eq:A_matrix_uniformDelay}
        A = \begin{bmatrix} I - c \Delta & 0 & \cdots & 0 & c H \\ I & 0 & \cdots & 0 & 0 \\ 0 & I & \ddots & \vdots & \vdots \\ \vdots & \vdots & \ddots & 0 & 0 \\ 0 & 0 & \cdots & I & 0 \end{bmatrix},
    \end{equation}
    where $A$ is $(d+1)n \times (d+1)n$. At consensus, $z^i(k) \to [\alpha, \alpha, \ldots, \alpha]^T$ (length $d+1$), and $A [\alpha, \alpha, \ldots, \alpha]^T = [\alpha, \alpha, \ldots, \alpha]^T$, since $(I - c \Delta) \alpha + c H \alpha = \alpha$.

    % Using the left eigenvector from the lemma
    From Lemma~\ref{lem:conservation_nonuniform}, the left eigenvector for the non-uniform delay case is expressed as in \eqref{eq:Eigenvector}, satisfying $V^T A = V^T$. For uniform delay $\tau_{ij} = d$, all edges have delay $d$, so $\tau_{\max} = d$, and $H_l = 0$ for $l \neq d$, $H_d = H$. Thus, the sums simplify to $\mathbf{1}^T \sum_{l=k}^{\tau_{\max}} H_l = \mathbf{1}^T H = [\deg(1), \ldots, \deg(n)]$ for $k \leq d$, and 0 otherwise. The eigenvector becomes
    \begin{align*}
        V^T = [ \beta \mathbf{1}^T, c \beta [\deg(1), \ldots, \deg(n)], \ldots,& \\
        c \beta [\deg(1), \ldots, \deg(n)]]&,
    \end{align*}
    with $d$ copies of $c \beta [\deg(1), \ldots, \deg(n)]$. The lemma establishes that $V^T z^i(k) = V^T z^i(0)$ for all $k \geq 0$. This conserved quantity is computed to find $\alpha_j$.

    % Conserved quantity
    Compute $V^T z^i(k)$ as follows:
    \begin{equation*}
        V^T z^i(k) = \beta \mathbf{1}^T w^i(k) + c \beta \sum_{m=1}^d [\deg(1), \ldots, \deg(n)] w^i(k-m).
    \end{equation*}
    At consensus, $w^i(k-m) \to \alpha$, and
    \begin{align} \label{eq:consensus_value1}
        V^T [\alpha, \alpha, \ldots, \alpha]^T &= \beta \mathbf{1}^T \alpha + c \beta \sum_{m=1}^d [\deg(1), \ldots, \deg(n)] \alpha \notag \\
        &= \beta n \alpha_j + c \beta d \alpha_j \sum_{i=1}^n \deg(i) \notag \\
        &= \beta \alpha_j (n + 2 c d |E|),
    \end{align}
    using the handshaking lemma $\sum_{i=1}^n \deg(i) = 2 |E|$. Initially, $w^i(-m) = w^i(0)$ for $m=1,\ldots,d$ and we have
    \begin{align} \label{eq:consensus_value2}
        V^T z^i(0) &= \beta \mathbf{1}^T w^i(0) + c \beta \sum_{m=1}^d [\deg(1), \ldots, \deg(n)] w^i(0) \notag \\
        &= \beta \sum_{j=1}^n w_j^i(0) + c \beta d \sum_{i=1}^n \deg(i) w_j^i(0).
    \end{align}
    From \eqref{eq:consensus_value1} and \eqref{eq:consensus_value2}, we have:
    \begin{equation*}
        \beta \alpha_j (n + 2 c d |E|) = \beta \sum_{i=1}^n (1 + c d \deg(i)) w_j^i(0),
    \end{equation*}
    \begin{equation*}
        \alpha_j = \frac{\sum_{i=1}^n (1 + c d \deg(i)) w_j^i(0)}{n + 2 c d |E|},
    \end{equation*}
    thus proving \eqref{eqConcensusValueUniform}.

    % Weight sum
    The weights are $\frac{1 + c d \deg(i)}{n + 2 c d |E|}$, which are positive since $ c > 0 $, $ \deg(i) \geq 1 $ (as $ G $ is connected), and $ n + c \cdot 2 |E| > 0 $ (with $ |E| \geq n - 1 $). Their sum is:
    \begin{equation*}
        \frac{\sum_{i=1}^n (1 + c d \deg(i))}{n + 2 c d |E|} = \frac{n + c d \cdot 2 |E|}{n + 2 c d |E|} = 1.
    \end{equation*}

    % Non-uniformity
    For $\alpha_j \neq \frac{1}{n} \sum_{i=1}^n w_j^i(0)$, compute
    \begin{align*}
        &\alpha_j - \frac{1}{n} \sum_{i=1}^n w_j^i(0) = \frac{\sum_{i=1}^n (1 + c d \deg(i)) w_j^i(0)}{n + 2 c d |E|} - \frac{\sum_{i=1}^n w_j^i(0)}{n} \\
        &= \frac{1}{n(n + 2 c d |E|)} \Bigg( n \sum_{i=1}^n (1 + c d \deg(i)) w_j^i(0) \\
        &- (n + 2 c d |E|) \sum_{i=1}^n w_j^i(0) \Bigg) \\
        &= \frac{c d}{n + 2 c d |E|} \left( \sum_{i=1}^n \deg(i) w_j^i(0) - \frac{2 |E|}{n} \sum_{i=1}^n w_j^i(0) \right),
    \end{align*}
    where $ \frac{2 |E|}{n} = \bar{d} $ (mean degree). This is zero if
    \begin{equation*}
        \sum_{i=1}^n \deg(i) w_j^i(0) = \bar{d} \sum_{i=1}^n w_j^i(0),
    \end{equation*}
    which holds when $ \deg(i) = \bar{d} $ for all $ i $ (regular graph), otherwise it is non-zero due to degree variation.

    % Bounds
    For the last part of the theorem, since $ \frac{1 + c d \deg(i)}{n + 2 c d |E|} > 0 $ and $ \sum_{i=1}^n \frac{1 + c d \deg(i)}{n + 2 c d |E|} = 1 $, $ \alpha_j $ is a convex combination \footnote{A convex combination is a linear combination of points where all coefficients are non-negative and add up to 1 \cite{rockafellar1997convex}.}. Let $ m_j = \min_i w_j^i(0) $ and $ M_j = \max_i w_j^i(0) $. Then
    \begin{align*}
        \alpha_j &\geq \sum_{i=1}^n \frac{1 + c d \deg(i)}{n + 2 c d |E|} m_j = m_j, \\
        \alpha_j &\leq \sum_{i=1}^n \frac{1 + c d \deg(i)}{n + 2 c d |E|} M_j = M_j,
    \end{align*}
    thus $ m_j \leq \alpha_j \leq M_j $, completing the proof.
\end{proof}

\subsection{Nonuniform delays}
Building on the uniform delay case, the next lemma and theorem extend the analysis to the general scenario of nonuniform delays, where $ \tau_{ij} $ varies across edges, addressing the challenges posed by heterogeneous communication lags in practical networks.

\begin{lemma}[Lyapunov Stability with LMI] \label{lem:Lyapunov_Stability_with_LMI}
    Consider the MAS in Lemma~\ref{lem:vector_update_nonuniform_delay}. The dynamics are
    \begin{equation}  \notag
        w_j(k+1) = (I - c \Delta) w_j(k) + c \sum_{m=1}^{\tau_{\max}} H_m w_j(k - m),
    \end{equation}
    %where $ H = \sum_{m=1}^{\tau_{\max}} H_m $, and 
    where $ L \mathbf{1} = 0 $. 
    The augmented state $ z(k) \in \mathbb{R}^{n (\tau_{\max} + 1)} $ is
    \begin{equation} \notag
        z(k) = [w_j(k)^T, w_j(k-1)^T, \ldots, w_j(k - \tau_{\max})^T]^T ,
    \end{equation}
    with dynamics
    \[
    z(k+1) = A z(k),
    \]
    \begin{equation} \label{eq:A_matrix_nonuniformDelay}
        A = 
        \begin{bmatrix} 
            I - c \Delta & c H_1 & c H_2 & \cdots & c H_{\tau_{\max}} \\ 
            I & 0 & 0 & \cdots & 0 \\ 
            0 & I & 0 & \cdots & 0 \\ 
            \vdots & \vdots & \ddots & \ddots & \vdots \\ 
            0 & 0 & \cdots & I & 0 
        \end{bmatrix}.
    \end{equation}
    The system converges to consensus, i.e., \( w_j^i(k) \to \alpha_j \) for all \( i \) as \( k \to \infty \) and for some \( \alpha_j \in \mathbb{R} \), if there exists a positive definite matrix \( P \in \mathbb{R}^{n (\tau_{\max} + 1) \times n (\tau_{\max} + 1)} \) satisfying the LMI
    \begin{equation*}
        A^T Q P Q A - Q P Q < 0,
    \end{equation*}
    where \( Q \in \mathbb{R}^{n (\tau_{\max} + 1) \times n (\tau_{\max} + 1)} \), defined as 
    \[ 
        Q = I - \frac{1}{n (\tau_{\max} + 1)} \mathbf{1}_{n (\tau_{\max} + 1)} \mathbf{1}_{n (\tau_{\max} + 1)}^T, 
    \]
    projects onto the subspace orthogonal to the consensus direction. A candidate \( P \) is given by 
    \[
        P = \sum_{d=0}^{\tau_{\max}} C_d (L + \delta I_n) C_d^T 
    \]
    with \( \delta > 0 \) small to ensure \( P \) is positive definite and $ C_d = [0, \ldots, I, \ldots, 0] $ with $ I $ in the $ (d+1) $-th block and zeros elsewhere. 
    
    This LMI provides a sufficient condition for consensus convergence and serves as a numerical feasibility test for a given step size \( c \). Although it does not yield an explicit analytical bound on \( c \), it can be used to compute the largest admissible \( c \) numerically.
\end{lemma}

\begin{proof}
    See Appendix II.
\end{proof}

\begin{theorem} \label{thm:delay_consensus_nonuniform}
    Consider the MAS in Lemma~\ref{lem:Lyapunov_Stability_with_LMI}, where the system converges to a consensus value $ \alpha_j $. The consensus value is given by
    \begin{equation} \label{eqConcensusValueNonuniform}
        \alpha_j = \frac{\sum_{i=1}^n \left(1 + c \sum_{j \in \mathcal{N}_i} \tau_{ij}\right) w^i(0)}{n + 2 c \sum_{(i,j) \in E} \tau_{ij}},
    \end{equation}
    where $ \sum_{j \in \mathcal{N}_i} \tau_{ij} $ is the total delay incoming to agent $ i $, and $ \sum_{(i,j) \in E} \tau_{ij} $ is the total delay over all edges (counting each edge once due to undirectedness). This is a weighted average of initial conditions with weights $ \frac{1 + c \sum_{j \in \mathcal{N}_i} \tau_{ij}}{n + c \sum_{(i,j) \in E} \tau_{ij}} $, which sum to 1. Then, $ \alpha_j \neq \frac{1}{n} \sum_{j=1}^n w_j^i(0) $ (the uniform average) unless $ \sum_{j \in \mathcal{N}_i} \tau_{ij} $ is equal for all $ i $, and the consensus value satisfies $ \min_i w_j^i(0) \leq \alpha_j \leq \max_i w_j^i(0) $.
\end{theorem}

\begin{proof}
    At steady state, the system converges to consensus, where $ W(k) \to W^* = \mathbf{1} \alpha^T $, with $ \alpha^T = [\alpha_1, \ldots, \alpha_n] $ representing the consensus values across all agents and components. Substituting this into the update rule $ W(k+1) = (I - c \Delta) W(k) + c \sum_{m=1}^{\tau_{\max}} H_m W(k - m) $, and noting that at steady state all delayed states also equal $ W^* $ (since $ W(k - m) \to W^* $ for all $ m $), one can obtain
    \begin{equation*}
        W^* = (I - c \Delta) W^* + c \sum_{m=1}^{\tau_{\max}} H_m W^*.
    \end{equation*}
    %and $ \Delta \mathbf{1} = H \mathbf{1} $.
    
    To determine $ \alpha $, define the augmented state vector $ z^i(k) = [w^i(k)^T, w^i(k-1)^T, \ldots, w^i(k - \tau_{\max})^T]^T \in \mathbb{R}^{n (\tau_{\max} + 1)} $, which stacks the current and all delayed priority vectors for agent $ i $. The dynamics are governed by
    \begin{equation} \label{eq:augmented}
        z^i(k+1) = A z^i(k),
    \end{equation}
    where the state transition matrix $ A $ is as in \eqref{eq:A_matrix_nonuniformDelay}.
    
    From Lemma~\ref{lem:conservation_nonuniform}, the vector $ V $ in \eqref{eq:Eigenvector} is a left eigenvector of $A$ with eigenvalue 1, satisfying $V^T A = V^T$. The lemma also establishes that the quantity $V^T z^i(k)$ is conserved, i.e., $V^T z^i(k) = V^T z^i(0)$ for all $k \geq 0$. This conserved quantity is computed to find the consensus value $\alpha_j$.
 
    The expression for $V^T z^i(k)$ from the lemma is
    \begin{equation} \label{eq:conserved_quantity}
        V^T z^i(k) = \beta \mathbf{1}^T w^i(k) + \beta c \sum_{m=1}^{\tau_{\max}} \left( \mathbf{1}^T \sum_{l=m}^{\tau_{\max}} H_l \right) w^i(k-m).
    \end{equation}
    At consensus, $w^i(k) \to \alpha_j \mathbf{1}$ for all $i$ and all delayed states $w^i(k - m) \to \alpha_j \mathbf{1}$ (since the system has converged and delays are finite). Substituting this into \eqref{eq:conserved_quantity}, we have
    \[
        V^T z^i(k) \to \beta \mathbf{1}^T (\alpha_j \mathbf{1}) + \beta c \sum_{m=1}^{\tau_{\max}} \left( \mathbf{1}^T \sum_{l=m}^{\tau_{\max}} H_l \right) (\alpha_j \mathbf{1}).
    \]
    Since $\mathbf{1}^T \mathbf{1} = n$, the first term is $ \beta \mathbf{1}^T (\alpha_j \mathbf{1}) = \beta n \alpha_j $. 
    The sum in the above can be rewritten as
    \[
        \sum_{m=1}^{\tau_{\max}} \left( \mathbf{1}^T \sum_{l=m}^{\tau_{\max}} H_l \right) \mathbf{1} = \sum_{m=1}^{\tau_{\max}} \sum_{l=m}^{\tau_{\max}} \mathbf{1}^T H_l \mathbf{1}.
    \]
    Since $\mathbf{1}^T H_l \mathbf{1} = \sum_{i=1}^{n} \sum_{j=1}^{n} (H_l)^i_j = 2 \sum_{(i,j) \in E: \tau_{ij} = l} 1$ (each edge with delay $l$ is counted twice in an undirected graph), the summation can be reordered as follows: 
    \footnote{The equality $\sum_{m=1}^{\tau_{\max}} \sum_{l=m}^{\tau_{\max}} \mathbf{1}^T H_l \mathbf{1} = \sum_{l=1}^{\tau_{\max}} l \cdot \mathbf{1}^T H_l \mathbf{1}$ holds as follows: 
    \begin{itemize}
        \item The left-hand side sums $\mathbf{1}^T H_l \mathbf{1}$ over $l$ from $m$ to $\tau_{\max}$ for each $m$ from 1 to $\tau_{\max}$, counting $\mathbf{1}^T H_1 \mathbf{1}$ once, $\mathbf{1}^T H_2 \mathbf{1}$ twice, ..., $\mathbf{1}^T H_{\tau_{\max}} \mathbf{1}$ $\tau_{\max}$ times.
        \item The right-hand side sums $l \cdot \mathbf{1}^T H_l \mathbf{1}$ over $l$ from 1 to $\tau_{\max}$, where $l$ weighs the edge count by the delay, matching the left-hand side’s counting. 
    \end{itemize}
    Thus, both sides are equal due to the reordering of the summation according to the alignment of the delay contributions.}:
    \begin{align*}
        \sum_{m=1}^{\tau_{\max}} \sum_{l=m}^{\tau_{\max}} \mathbf{1}^T H_l \mathbf{1} &= \sum_{l=1}^{\tau_{\max}} l \cdot \mathbf{1}^T H_l \mathbf{1} = 2 \sum_{l=1}^{\tau_{\max}} l \sum_{(i,j) \in E: \tau_{ij} = l} 1 \\
        &= 2 \sum_{(i,j) \in E} \tau_{ij}.
    \end{align*}
    The equality $ \sum_{l=1}^{\tau_{\max}} l \sum_{(i,j) \in E: \tau_{ij} = l} 1 = \sum_{(i,j) \in E} \tau_{ij} $ holds because the left-hand side computes the total delay across all edges. For each delay $ l $, $ \sum_{(i,j) \in E: \tau_{ij} = l} 1 $ counts edges with $ \tau_{ij} = l $, and multiplying by $ l $ gives their delay contribution. Summing over $ l $ from 1 to $ \tau_{\max} $, each edge $(i, j)$ contributes its delay $ \tau_{ij} $ exactly once, so $ \sum_{l=1}^{\tau_{\max}} l \sum_{(i,j) \in E: \tau_{ij} = l} 1 = \sum_{(i,j) \in E} \tau_{ij} $.

    Thus, $\sum_{m=1}^{\tau_{\max}} \left( \mathbf{1}^T \sum_{l=m}^{\tau_{\max}} H_l \right) \mathbf{1} = 2 \sum_{(i,j) \in E} \tau_{ij}$, and
    \begin{equation} \label{eq:VTzi_at_infinity}
        V^T z^i(k) \to \beta n \alpha_j + 2 \beta c \alpha_j \sum_{(i,j) \in E} \tau_{ij}.
    \end{equation}
    Initially, assuming $w^i(-m) = w^i(0)$ for $m = 1, \ldots, \tau_{\max}$, we have
    \[
        V^T z^i(0) = \beta \mathbf{1}^T w^i(0) + \beta c \sum_{m=1}^{\tau_{\max}} \left( \mathbf{1}^T \sum_{l=m}^{\tau_{\max}} H_l \right) w^i(0).
    \]
    Since $w^i(0)$ is the same for each $m$ under this assumption, we have
    \[
        \sum_{m=1}^{\tau_{\max}} \left( \mathbf{1}^T \sum_{l=m}^{\tau_{\max}} H_l \right) w^i(0) = \left(\sum_{m=1}^{\tau_{\max}} \left( \sum_{l=m}^{\tau_{\max}} \mathbf{1}^T H_l \right) \right) w^i(0).
    \]
    The first term is $ \beta \mathbf{1}^T w^i(0) = \beta \sum_{i=1}^{n} w^i(0) $. Define $\psi(i) = \sum_{j \in \mathcal{N}_i} \tau_{ij}$. Then, 
    \[
        \sum_{m=1}^{\tau_{\max}} \mathbf{1}^T \sum_{l=m}^{\tau_{\max}} H_l = [\psi(1), \ldots, \psi(n)], 
    \]
    and
    \begin{equation} \label{eq:VTzi_at_zero}
        V^T z^i(0) = \beta \sum_{i=1}^n w^i(0) + \beta c \sum_{i=1}^n \psi(i) w^i(0).
    \end{equation}
    Equating \eqref{eq:VTzi_at_infinity} and \eqref{eq:VTzi_at_zero} and solving yields
    \[
        \beta \sum_{i=1}^n w^i(0) + \beta c \sum_{i=1}^n \psi(i) w^i(0) = n \beta \alpha_j + 2 \beta c \alpha_j \sum_{(i,j) \in E} \tau_{ij}.
    \]
    Since $\beta \neq 0$, both sides can be divided by $\beta$, leading to
    \[
        \alpha_j = \frac{\sum_{i=1}^n \left(1 + c \psi(i)\right) w^i(0)}{n + 2 c \sum_{(i,j) \in E} \tau_{ij}}.
    \]
\end{proof}

\section{Simulations and Experiments} \label{sec:results}
\subsection{Simulations}
Before conducting real-world experiments, the proposed consensus algorithm is evaluated through numerical simulations. The kinematic model of the QBot3 robot, along with the relevant motion dynamics, terminology used on the QBot3 platform, and the control structure, are detailed in \cite{naderi2025distributed}.

Each simulated agent follows the same update rule as described in Section~\ref{sec:main_results}, interacting over a predefined communication graph that includes nonuniform constant delays. These simulations allow us to verify the theoretical convergence guarantees in a controlled environment, visualize the dynamic behavior of the agents, and investigate the influence of delay parameters and step size on convergence speed and final values.

For the simulation parameters and initial conditions, the same values as reported in \cite{naderi2025distributed} are adopted. %The controller gains are set to $k_p = 1.25$ and $k_d = 1$. 
The communication delays between agents are modeled by the delay matrix
\begin{equation} \notag
    \tau = 
    \begin{bmatrix} 
    0 & 7 & 1 & 5 \\ 
    7 & 0 & 5 & 5 \\  
    1 & 5 & 0 & 6 \\  
    5 & 5 & 6 & 0 \\ 
    \end{bmatrix}.
\end{equation}

Figure~\ref{fig:convergence_c} illustrates the time evolution of the priority matrix columns for 4 agents over 4 seconds under three different step sizes. Each subplot shows the evolution of one column of $W$, with the curves corresponding to the values held by the four agents. As observed in the simulations and consistent with the analysis, the step size $c$ exhibits a non-monotonic effect on convergence speed. A small step size ($c=0.05$) yields smooth, monotonic convergence but requires a long settling time. An intermediate value ($c=0.25$) provides the fastest convergence with only mild oscillations. A larger step size ($c=0.55$) significantly increases oscillatory behavior due to overcorrection based on delayed neighbor information, resulting in a long overall settling time despite faster initial correction. 

\begin{figure*}[tb]
    \centering
    \includegraphics[width=\textwidth]{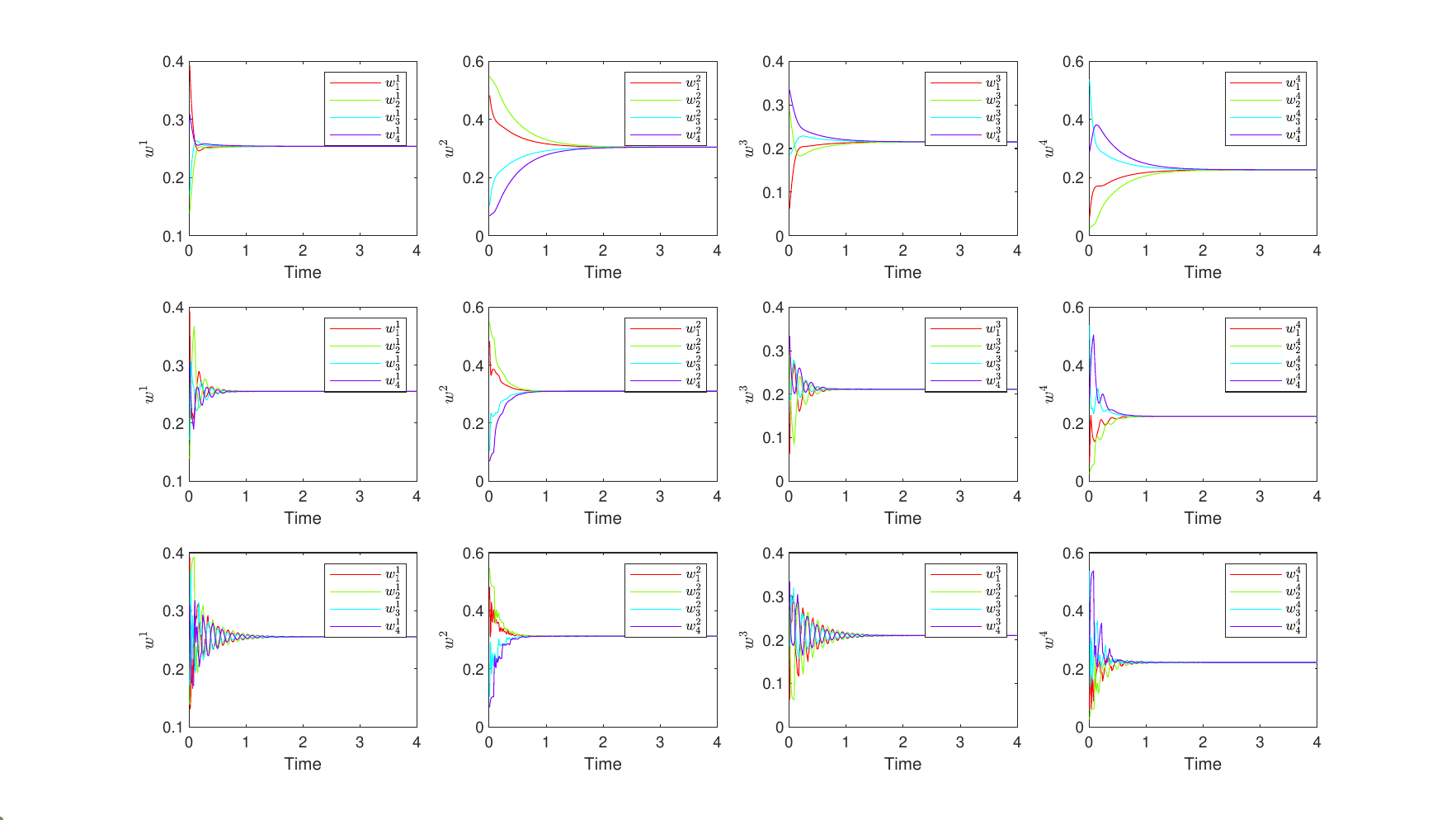}
    \caption{Time evolution of the priority matrix columns for 4 agents over 4 seconds. Each row corresponds to a different step size $c$: $c=0.05$ (first row), $c=0.25$ (second row), and $c=0.55$ (third row).}
    \label{fig:convergence_c}
\end{figure*}

Table~\ref{tab:convergence_summary} reports the convergence times for the different step sizes. Convergence time is defined as the first time instant when the maximum difference within each column falls below $10^{-4}$ and remains below thereafter. The values in the table show that an intermediate step size ($c=0.25$) yields the fastest convergence, while overly small or large values result in longer settling times due to slow adaptation or oscillatory behavior, respectively.

\begin{table}[t]
    \caption{Convergence time and final consensus values for different step sizes $c$ (4 agents, nonuniform delays).}
    \centering
    \begin{tabular}{c c c c c c}
        \hline
        $c$ & Convergence time (s) & $\alpha_1$ & $\alpha_2$ & $\alpha_3$ & $\alpha_4$ \\
        \hline
        0.05 & 3.65 & 0.2537 & 0.3046 & 0.2151 & 0.2266 \\
        0.15 & 1.78 & 0.2545 & 0.3083 & 0.2127 & 0.2245 \\
        0.25 & 1.43 & 0.2548 & 0.3102 & 0.2114 & 0.2235 \\
        0.35 & 1.82 & 0.2551 & 0.3113 & 0.2107 & 0.2229 \\
        0.45 & 2.26 & 0.2552 & 0.3120 & 0.2103 & 0.2225 \\
        0.55 & 2.69 & 0.2553 & 0.3125 & 0.2099 & 0.2223 \\
        \hline
    \end{tabular}
    \label{tab:convergence_summary}
\end{table}

To further validate the observed convergence behavior, a reduced form of the LMI from Lemma~\ref{lem:Lyapunov_Stability_with_LMI} was numerically solved. By projecting the dynamics onto the disagreement subspace (i.e., removing the consensus mode), numerical issues associated with the marginal eigenvalue at 1 were avoided, and a discrete-time Lyapunov LMI was solved on the resulting reduced system. A bisection search over the step size \( c \) was performed to determine the largest value ensuring feasibility of the strict LMI on the reduced system. This yielded a maximum admissible step size of \( c_{\max} = 0.58 \). This bound is consistent with the simulation results and provides a quantitative stability reference for the step-size selection reported in Table~\ref{tab:convergence_summary}. 
Figure~\ref{fig:divergence_c} illustrates the system behavior for a step size slightly above this bound. When $c = 0.59$, the trajectories no longer converge to consensus and instead exhibit growing oscillations, indicating divergence of the disagreement dynamics. This observation agrees with the LMI-based analysis and confirms that step sizes exceeding the admissible limit lead to loss of convergence.
\begin{figure*}[hbt]
    \centering
    \includegraphics[width=\textwidth]{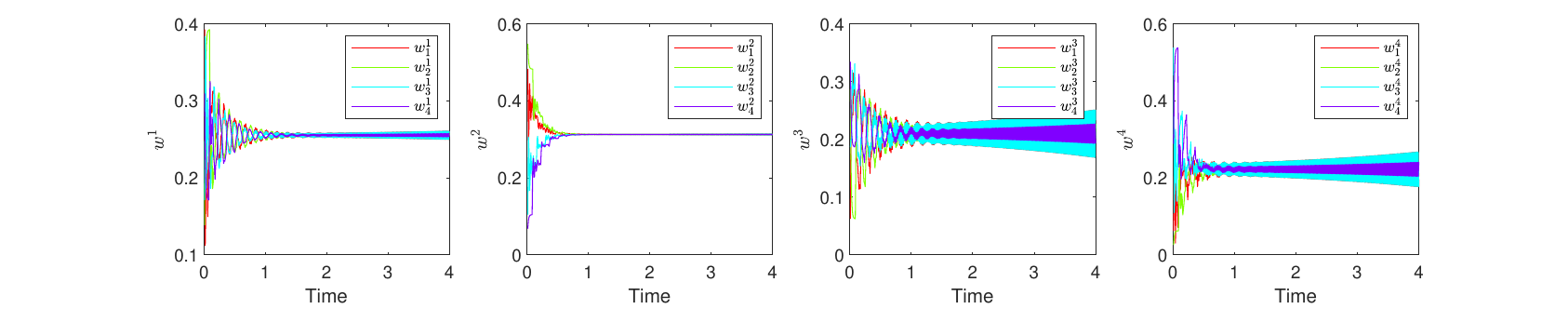}
    \caption{Time evolution of the priority matrix columns for $c=0.59$, slightly above the LMI-certified bound $c_{\max}=0.58$.}
    \label{fig:divergence_c}
\end{figure*}

\subsection{Experiments}
To validate the theoretical convergence results in a real-world setting, the proposed consensus algorithm is implemented on a swarm of four Quanser QBot3 mobile ground robots. Each QBot3 is equipped with onboard computing, differential-drive motors, and wireless communication capabilities, enabling decentralized interaction in physical environments.

The algorithm was deployed using the QUARC Simulink library, which interfaces with MATLAB/Simulink to compile and run control code directly on the robots. Robot positions were tracked using a motion capture system composed of OptiTrack infrared cameras and Motive software, which provided estimates of each robot's position and orientation. These data were streamed into the Simulink environment for logging and evaluation. 

In the experiments, each robot is driven by a simple tracking controller that converts the consensus reference points generated by the algorithm into linear and angular velocity commands. The linear velocity is determined from the distance to the reference using a bounded sigmoid mapping, while the angular velocity is generated by a proportional–derivative controller acting on the heading error. 
The consensus update law and the controller were implemented in MATLAB/Simulink. At each simulation step, the consensus algorithm updated the rendezvous reference points for each agent. The control signals were then generated and applied to the QBot3's inverse kinematic model. 

While no artificial or predefined delays were introduced in the implementation, delays naturally arise in the physical system due to factors such as communication latency, asynchronous processing times, and Wi-Fi transmission variability. These real-time effects result in nonuniform delays across agents, even though the algorithm itself does not explicitly incorporate or measure them. The exact values of these delays are unknown during experiments, preventing us from computing the theoretical convergence value for direct comparison. Nonetheless, experimental observations confirm that the agents’ priority vectors converge consistently, demonstrating the algorithm's robustness under practical communication conditions.

\begin{figure*}[htp]
    \centering
    \includegraphics[width=\textwidth]{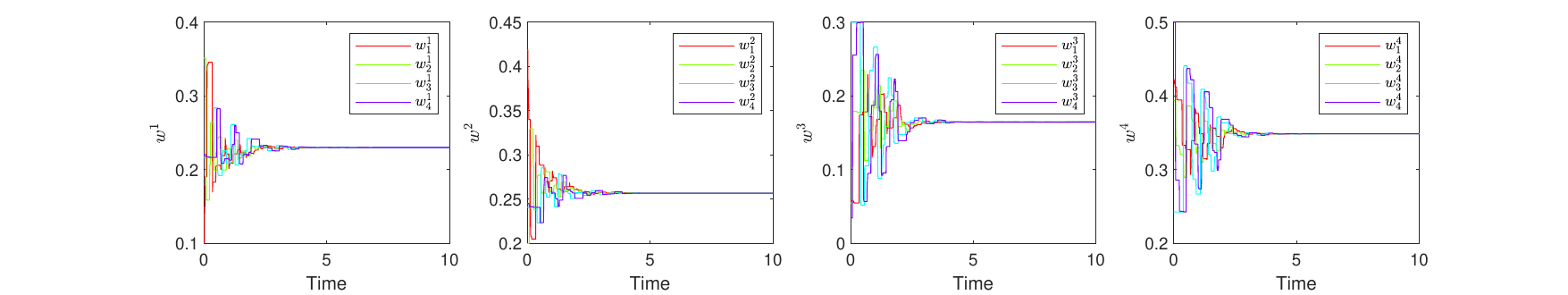}
    \caption{Experimental time evolution of the priority matrix columns for $c=0.20$.}
    \label{fig:W_experiment}
\end{figure*}

Figure~\ref{fig:W_experiment} illustrates the experimental time evolution of the priority matrix columns for $c=0.20$. The agents start from the randomly generated initial priority matrix
\[
W(0)=
\begin{bmatrix}
0.1000 & 0.4189 & 0.0587 & 0.4224 \\
0.3513 & 0.2000 & 0.0548 & 0.3939 \\
0.2165 & 0.2407 & 0.3000 & 0.2427 \\
0.2205 & 0.2451 & 0.0344 & 0.5000
\end{bmatrix}.
\]
Despite these different initial values, all trajectories converge to a common value, demonstrating consensus among the agents. The transient oscillations observed in the first seconds are consistent with the delayed update dynamics predicted by the analysis and gradually vanish as the system approaches the steady state.

\section{Conclusion} \label{sec:conclusion}
This paper analyzed the convergence of a distributed multi-objective optimization algorithm under constant communication delays in MASs. For the uniform delay case, Rouché’s theorem was used to derive exact stability conditions, while a Lyapunov-based approach was employed for the nonuniform delay case. These analyses provide rigorous convergence guarantees and explicit expressions for the consensus value and admissible step-size bounds. Simulation results and experiments on QBot3 robots validated the theoretical findings and demonstrated robustness to communication delays.

Future work will consider extensions to time-varying delays and fully asynchronous update schemes, in which agents operate on different clocks or update at different rates, beyond the communication-level asynchrony considered in this work.

\appendices

\section{Proving the Convergence of Lemma~\ref{lem:convergence_uniform_delay}}
    For each component $ j $, the columns $ w_j(k) = [w^1_j(k), \dots, w^n_j(k)]^T $ have the structure
    \[
        w_j(k+1) = (I - c \Delta) w_j(k) + c H w_j(k - d).
    \]
    Define the augmented state $ z_j(k) = [w_j(k)^T, w_j(k-1)^T, \ldots, w_j(k-d)^T]^T \in \mathbb{R}^{n(d+1)} $, so the dynamics become $ z_j(k+1) = A z_j(k) $, with $A$ having the structure given in \eqref{eq:A_matrix_uniformDelay}. Convergence to consensus requires that the eigenvalue $ \lambda = 1 $, corresponding to the consensus state, be on the unit circle ($ |\lambda| = 1 $) and that all other eigenvalues of $ A $ be of magnitude less than 1 ($|\lambda| < 1 $). The characteristic polynomial is
    \[
        \lambda_A I - A = \begin{bmatrix}
        (\lambda_A - 1) I_n + c \Delta & 0 & \cdots & 0 & -c H \\
        -I & \lambda_A I_n & \cdots & 0 & 0 \\
        0 & -I & \cdots & 0 & 0 \\
        \vdots & \vdots & \ddots & \vdots & \vdots \\
        0 & 0 & \cdots & -I & \lambda_A I_n
        \end{bmatrix}.
    \]
    Using the block determinant formula for a matrix $ \begin{bmatrix} A_{11} & A_{12} \\ A_{21} & A_{22} \end{bmatrix} $, if $ A_{22} $ is invertible, $ \det = \det(A_{22}) \det(A_{11} - A_{12} A_{22}^{-1} A_{21}) $. Here, $ A_{11} = (\lambda_A - 1) I_n + c \Delta $ ($ n \times n $), $ A_{12} = [0, \cdots, 0, -c H] $ ($ n \times nd $), $ A_{21} = [-I, 0, \cdots, 0 ]^T $ ($ nd \times n $), and $ A_{22} = \begin{bmatrix} \lambda_A I_n & 0 & \cdots & 0 \\ -I & \lambda_A I_n & \cdots & 0 \\ \vdots & \ddots & \ddots & \vdots \\ 0 & \cdots & -I & \lambda_A I_n \end{bmatrix} $ ($ nd \times nd $). 

    The matrix $ A_{22} $ is lower-triangular with $ \lambda_A I_n $ on the diagonal ($ d $ times). Thus $\det(A_{22}) = \det(\lambda_A I_n)^d = \lambda_A^{nd} $.
    The inverse $ A_{22}^{-1} $ (with $ n \times n $ blocks) is upper-triangular:
    \[
        A_{22}^{-1} = \begin{bmatrix}
        \lambda_A^{-1} I_n & 0 & \cdots & 0 \\
        \lambda_A^{-2} I_n & \lambda_A^{-1} I_n & \cdots & 0 \\
        \vdots & \ddots & \ddots & \vdots \\
        \lambda_A^{-d} I_n & \lambda_A^{-d+1} I_n & \cdots & \lambda_A^{-1} I_n
        \end{bmatrix},
    \]
    which is obtained by solving $ A_{22} A_{22}^{-1} = I $. Then
    \[
        A_{22}^{-1} A_{21} = [-\lambda_A^{-1} I_n, -\lambda_A^{-2} I_n, \cdots, -\lambda_A^{-d} I_n]^T,
    \]
    \[
        A_{12} A_{22}^{-1} A_{21} = [0, \cdots, 0, -c H] \begin{bmatrix} -\lambda_A^{-1} I_n \\ -\lambda_A^{-2} I_n \\ \vdots \\ -\lambda_A^{-d} I_n \end{bmatrix} = c \lambda_A^{-d} H,
    \]
    \[
        A_{11} - A_{12} A_{22}^{-1} A_{21} = (\lambda_A - 1) I_n + c \Delta - c \lambda_A^{-d} H.
    \]
    The determinant is
    \begin{align} \notag
        \det(\lambda_A I - A) &= \det(A_{22}) \det(A_{11} - A_{12} A_{22}^{-1} A_{21})\\ \notag
        &= \lambda_A^{nd} \det\left( (\lambda_A - 1) I_n + c \Delta - c \lambda_A^{-d} H \right) \\ \notag
        &= \det\left( \lambda_A^{d}(\lambda_A - 1) I_n + c \lambda_A^{d} \Delta - c H \right)\\ \notag
        &= \det\left( \lambda_A^{d+1} I_n - \lambda_A^d (I_n - c \Delta) - c H \right).
    \end{align}
    By substituting $ H = \Delta - L $, define $ D(\lambda_A) := (\lambda_A^{d+1} - \lambda_A^d) I + c (\lambda_A^d \Delta - \Delta + L) $.

    The eigenvalues of $ L $ satisfy $ L v_{i,L} = \lambda_{i,L} v_{i,L} $, where $ 0 = \lambda_{1,L} < \lambda_{2,L} \leq \cdots \leq \lambda_{n,L} \leq 2 \Delta_{\max} $. Moreover, since $ L $ is symmetric, it has an orthogonal eigenvector matrix $ V = [v_{1,L}, \ldots, v_{n,L}] $, with $ V^T V = I $, meaning the eigenvectors are orthonormal: each eigenvector $ v_{i,L} $ is normalized ($ v_{i,L}^T v_{i,L} = 1 $, i.e., $ \|v_{i,L}\|_2 = 1 $) and orthogonal to the others ($ v_{i,L}^T v_{j,L} = 0 $ for $ i \neq j $). Specifically, for the first eigenvector, $ v_{1,L} = \frac{1}{\sqrt{n}} \mathbf{1} $, which is the normalized form of $ \mathbf{1} $ (since $ L \mathbf{1} = 0 $, and $ \|\mathbf{1}\|_2 = \sqrt{n} $, so $ \frac{1}{\sqrt{n}} \mathbf{1} $ has norm 1). For any matrix $ D $, if $ V $ is orthogonal, $ \det(V^T D V) = \det(D) $. Thus, $ \det(D(\lambda_A)) = \det(V^T D(\lambda_A) V) $, and
    \begin{align*}
        V^T D(\lambda_A) V = &(\lambda_A^{d+1} - \lambda_A^d) I \\
        +& c (\lambda_A^d V^T \Delta V - V^T \Delta V + V^T L V),
    \end{align*}
    where $ V^T L V = \text{diag}(\lambda_{1,L}, \ldots, \lambda_{n,L}) $ and it is assumed that $ V^T \Delta V = \text{diag}(\delta_1, \ldots, \delta_n) $, where $ \delta_i = v_{i,L}^T \Delta v_{i,L} $. This assumption states that the matrix $ V^T \Delta V $ has no off-diagonal elements. This holds when the graph is regular, i.e., all nodes have the same degree $ deg_i $, so $ \Delta = deg_i I $ and $ V^T \Delta V = deg_i I $; otherwise, $ V^T \Delta V $ generally has non-zero off-diagonal elements, though these terms are smaller when the degrees are more uniform (e.g., when the variance in node degrees is low), making the approximation $ V^T \Delta V \approx \text{diag}(\delta_1, \ldots, \delta_n) $ more accurate. The stability analysis focuses on the diagonal terms $ \delta_i $, as the characteristic equation is evaluated along each eigenspace of $ L $, leading to a decoupled per-mode analysis. 
    Thus
    \[
    V^T D(\lambda_A) V = \text{diag} \left( \lambda_A^{d+1} - \lambda_A^d + c (\lambda_A^d \delta_i - \delta_i + \lambda_{i,L}) \right)_{i=1}^n,
    \]
    \[
    \det(D(\lambda_A)) = \prod_{i=1}^n \left[ \lambda_A^{d+1} - (1 - c \delta_i) \lambda_A^d - c (\delta_i - \lambda_{i,L}) \right].
    \]
    The characteristic equation per each $ \lambda_{i,A} $ is
    \begin{equation} \label{eq:characteristic_equation_d}
        p_i(\lambda_{i,A}) = \lambda_{i,A}^{d+1} - (1 - c \delta_i) \lambda_{i,A}^d - c (\delta_i - \lambda_{i,L}) = 0.
    \end{equation}
    % Stability analysis for the consensus mode
    For stability, roots must satisfy $ |\lambda_{i,A}| < 1 $ except for the consensus mode. For $ i = 1 $, $ \lambda_{1,L} = 0 $, and it follows that
    \[
    p_1(\lambda_{1,A}) = \lambda_{1,A}^{d+1} - (1 - c \delta_1) \lambda_{1,A}^d - c \delta_1 = 0.
    \]
    %Check $ \lambda_{1,A} = 1 $: $ 1^{d+1} - (1 - c \delta_1) 1^d - c \delta_1 = 1 - (1 - c \delta_1) - c \delta_1 = 0 $, so $ \lambda_{1,A} = 1 $ is a root. The other $ d $ roots must have $ |\lambda_{1,A}| < 1 $. For small $ c $, approximate roots are near 0; their magnitude scales as $ (c \delta_1)^{1/d} $, so $ c < \frac{1}{\delta_1} $ ensures $ |\lambda_{1,A}| < 1 $ for these roots.
    Factoring $ p_1(\lambda_{1,A}) $ using synthetic division yields
    \[
        p_1(\lambda_{1,A}) = (\lambda_{1,A} - 1) (\lambda_{1,A}^d + c \delta_1 \lambda_{1,A}^{d-1} + \cdots + c \delta_1 \lambda_{1,A} + c \delta_1).
    \]
    Letting $ q(\lambda_{1,A}) = \lambda_{1,A}^d + c \delta_1 (\lambda_{1,A}^{d-1} + \cdots + \lambda_{1,A} + 1) $, and assuming a root $ |\lambda_{1,A}| \geq 1 $, we have
    \[
        \lambda_{1,A}^d = -c \delta_1 (\lambda_{1,A}^{d-1} + \cdots + \lambda_{1,A} + 1),
    \]
    %Take magnitudes:
    \[
        |\lambda_{1,A}|^d \leq c \delta_1 \sum_{k=0}^{d-1} |\lambda_{1,A}|^k \leq c \delta_1 d |\lambda_{1,A}|^{d-1} \leq c \delta_1 d.
    \]
    If $ c < \frac{1}{d \delta_1} $, then $ |\lambda_{1,A}| < 1 $, contradicting $ |\lambda_{1,A}| \geq 1 $. Thus, all roots of $ q(\lambda_{1,A}) $ satisfy $ |\lambda_{1,A}| < 1 $, so $ p_1(\lambda_{1,A}) $ has one root at $ \lambda_{1,A} = 1 $ and $ d $ roots inside $ |\lambda_{1,A}| < 1 $.

    Since $ L v_{i,L} = \lambda_i v_{i,L} $, 
    \begin{equation} \label{Bound_H}
        v_{i,L}^T (\Delta - H) v_{i,L} = \lambda_{i,L}
    \implies
        \delta_i - h_i = \lambda_{i,L}.
    \end{equation}
    Thus, For $ i > 1 $, consider $ p_i(\lambda_{i,A}) = \lambda_{i,A}^{d+1} - (1 - c \delta_i) \lambda_{i,A}^d - c h_i $, where $ h_i = \delta_i - \lambda_{i,L} = v_{i,L}^T H v_{i,L} > 0 $.  
    Rouché’s theorem can be applied on the unit circle $ |\lambda_A| = 1 $ to find the bound on $ c $ that ensures stability. Rouché’s theorem states that if $ h(\lambda_{i,A}) $ and $ g(\lambda_{i,A}) $ are analytic inside and on $ |\lambda_{i,A}| = 1 $, and $ |g(\lambda_{i,A})| < |h(\lambda_{i,A})| $ on $ |\lambda_{i,A}| = 1 $, then $ p_i(\lambda_{i,A}) = h(\lambda_{i,A}) + g(\lambda_{i,A}) $ has the same number of roots inside $ |\lambda_{i,A}| = 1 $ as $ h(\lambda_{i,A}) $. The goal is to choose $ h(\lambda_{i,A}) $ with all $ d+1 $ roots inside $ |\lambda_{i,A}| = 1 $, and ensure the magnitude of $ g(\lambda_{i,A}) $ is strictly less than that of $ h(\lambda_{i,A}) $ on the unit circle. 
    
    Consider that $ p_i(\lambda_{i,A}) = h(\lambda_{i,A}) + g(\lambda_{i,A}) $, where $ h(\lambda_{i,A}) = \lambda_{i,A}^{d+1} - (1 - c \delta_i) \lambda_{i,A}^d $, $ g(\lambda_{i,A}) = -c h_i $. Roots of $ h(\lambda_{i,A}) = \lambda_{i,A}^d (\lambda_{i,A} - (1 - c \delta_i)) $ are 0 ($ d $ times) and $ 1 - c \delta_i $, inside $ |\lambda_{i,A}| = 1 $ if $ c < \frac{2}{\delta_i} $. On $ |\lambda_{i,A}| = 1 $, $ \lambda_{i,A} = e^{i\theta} $, and 
    \begin{align} \notag
        |h(\lambda_{i,A})| &= |e^{i d \theta} (e^{i\theta} - (1 - c \delta_i))| = |e^{i\theta} - (1 - c \delta_i)| \\  \notag
        %&= 1 \cdot |e^{i\theta} - (1 - c \delta_i)| \\ \notag
        &= \sqrt{\sin^2\theta + \cos^2\theta + (1 - c \delta_i)^2 - 2 \cos\theta (1 - c \delta_i)} \\ \notag
        &= \sqrt{2 + c^2 \delta_i^2 - 2 c \delta_i - 2 \cos\theta (1 - c \delta_i)} \\ \notag
        &= \sqrt{(1 - c \delta_i) (2 - 2 \cos\theta) + c^2 \delta_i^2} \\ \notag
        &= \sqrt{2 (1 - \cos\theta) (1 - c \delta_i) + (c \delta_i)^2},
    \end{align}
    which ranges from $ c \delta_i $ (at $ \cos\theta = 1 $) to $ \sqrt{4 (1 - c \delta_i) + (c \delta_i)^2} $ (at $ \cos\theta = -1 $). Meanwhile, $ |g(\lambda_{i,A})| = c h_i $. The minimum of $ |h(\lambda_{i,A})| $ is $ c \delta_i $, and since $ \delta_i > h_i = \delta_i - \lambda_{i,L} $ for $ i > 1 $ (as $ \lambda_{i,L} > 0 $), we have $ c \delta_i > c h_i $, ensuring $ |h(\lambda_{i,A})| > |g(\lambda_{i,A})| $ everywhere on the circle for $ i > 1 $. Rouché’s theorem applies if $ |g| < |h| $ on $ |\lambda_{i,A}| = 1 $, which is satisfied. Since $ h(\lambda_{i,A}) $ has $ d+1 $ roots inside $ |\lambda_{i,A}| < 1 $ when $ c < \frac{2}{\delta_i} $, $ p_i(\lambda_{i,A}) $ has all $ d+1 $ roots inside the unit circle.
    
    For $ i = 1 $, the inequality fails ($ \delta_1 = v_{1,L}^T H v_{1,L} $) since $ \lambda_{1,L} = 0 $, but we do not apply Rouché’s theorem here, as the roots of $ p_1(\lambda_{1,A}) $ were computed directly (one root at $ \lambda_{1,A} = 1 $ and $ d $ roots inside $ |\lambda_{1,A}| < 1 $), and the consensus mode is handled separately to ensure $ \lambda_{1,A} = 1 $ is on the unit circle while $ c d \delta_1 < 1 $. 

    To find the bound on $ \delta_i $, note that since $ V^T V = I $, we have $ \sum_{j=1}^n (v_{i,L}(j))^2 = 1 $, and since $ \delta_i = v_{i,L}^T \Delta v_{i,L} = \sum_{j=1}^n \deg(j) (v_{i,L}(j))^2 $, we have $ 0 \leq \delta_i \leq \Delta_{\max} $, as $ \deg(j) \leq \Delta_{\max} $.

    Combining conditions, for $ i = 1 $, $ c < \frac{1}{d \delta_1} $ ensures one root at 1 and others inside the unit circle. For $ i > 1 $, $ c < \frac{2}{\delta_i} $, and since $ \delta_i \leq \Delta_{\max} $, $ c < \frac{2}{\Delta_{\max}} $ suffices. Thus, the system converges to consensus if $ 0 < c < \min\left( \frac{1}{d \delta_1}, \frac{2}{\Delta_{\max}} \right) $.

\section{Proving the Convergence of Lemma~\ref{lem:Lyapunov_Stability_with_LMI}}
    Consider the following Lyapunov-Krasovskii function for each component $ j $:
    \begin{equation}
        V_j(k) = z(k)^T Q P Q z(k),
    \end{equation}
    with total Lyapunov function given by $ V(k) = \sum_{j=1}^m V_j(k) $.
    The Lyapunov stability conditions are verified analytically using the proposed Lyapunov function. Specifically, it is shown that the function is zero at consensus, prove the positive definiteness of $V_j(k)$ in the disagreement subspace, and show that the Lyapunov difference $\Delta V(k)$ is negative.
    
    At consensus, where all agents' states are equal (i.e., $ w_j^i(k - d) = \alpha_j $ for all $ i $ and $ d = 0, \ldots, \tau_{\max} $), the augmented state is
    \begin{align*}
        z(k) &= [w_j(k)^T, w_j(k-1)^T, \ldots, w_j(k - \tau_{\max})^T]^T \\ \notag
        &= [\alpha_j \mathbf{1}^T, \alpha_j \mathbf{1}^T, \dots, \alpha_j \mathbf{1}^T] = \alpha_j \mathbf{1}_{n (\tau{\max} + 1)}.
    \end{align*}
    $ Q z(k) = 0 $ because $ z(k) $ lies in the consensus direction, and its projection onto the disagreement subspace would result in zero. This occurs because $ Q = I - \frac{1}{n (\tau_{\max} + 1)} \mathbf{1}_{n (\tau{\max} + 1)} \mathbf{1}_{n (\tau{\max} + 1)}^T $ projects any vector onto the subspace orthogonal to $ \mathbf{1}_{n (\tau{\max} + 1)} $. At consensus, 
    \[
        Q z(k) = Q (\alpha_j \mathbf{1}_{n (\tau{\max} + 1)}) = \alpha_j Q \mathbf{1}_{n (\tau{\max} + 1)}.
    \]
    Developing $ Q \mathbf{1}_{n (\tau{\max} + 1)} $, we have
    \begin{align*}
        &Q \mathbf{1}_{n (\tau{\max} + 1)} = \mathbf{1}_{n (\tau{\max} + 1)}\\
        &- \frac{1}{n (\tau_{\max} + 1)} \mathbf{1}_{n (\tau{\max} + 1)} \mathbf{1}_{n (\tau{\max} + 1)}^T \mathbf{1}_{n (\tau{\max} + 1)} \\ \notag
        &= \mathbf{1}_{n (\tau{\max} + 1)} - \frac{1}{n (\tau_{\max} + 1)} \cdot n (\tau_{\max} + 1) \mathbf{1}_{n (\tau{\max} + 1)} \\ \notag 
        &= \mathbf{1}_{n (\tau{\max} + 1)} - \mathbf{1}_{n (\tau{\max} + 1)} = 0,
    \end{align*}
    so $ Q z(k) = 0 $. Thus $ V_j(k) = z(k)^T Q P Q z(k) = 0 $, ensuring $ V(k) = 0 $.
    
    To prove the positive definiteness, consider $ V_j(k) = z(k)^T Q P Q z(k) $, with 
    \begin{equation} \label{eq:P_definition}
        P = \sum_{d=0}^{\tau_{\max}} C_d (L + \delta I_n) C_d^T , 
    \end{equation}
    and $ \delta $ a positive scalar chosen to ensure $ P $ is positive definite. To show $ Q P Q $ is positive definite in the orthogonal subspace, note that $ L + \delta I_n $ has eigenvalues shifted by $ \delta $, making it positive definite. The matrix $ C_d $ extracts the $ d $-th block, so $ C_d (L + \delta I_n) C_d^T $ places $ L + \delta I_n $ in the $ (d+1, d+1) $-th block, and summing gives $ P = \text{diag}(L + \delta I_n, \ldots, L + \delta I_n) $. For $ z \neq 0 $, 
    \[ z^T P z = \sum_{d=0}^{\tau_{\max}} z_d^T (L + \delta I_n) z_d > 0 
    \]
    unless all $ z_d $ are proportional to $ \mathbf{1} $. Applying $ Q $, $ Q z = 0 $ if $ z = \alpha \mathbf{1}_{n (\tau_{\max} + 1)} $, and for $ Q z \neq 0 $, $ z^T Q P Q z > 0 $ since $ P > 0 $ and $ Q $ projects onto a subspace where $ L z_d \neq 0 $ unless $ z_d = \alpha \mathbf{1} $, indicating disagreement among agents. Thus, $ V_j(k) > 0 $ for $ Q z(k) \neq 0 $, and $ V(k) > 0 $. 
    
    To prove the negative difference, compute 
    \begin{align*} 
        V_j(k+1) &= z(k+1)^T Q P Q z(k+1) \\ \notag 
        &= (A z(k))^T Q P Q (A z(k)) \\ \notag
        &= z^T(k) A^T Q P Q A z(k),
    \end{align*} 
    so the difference $ \Delta V_j(k) $ is
    \begin{align*}
        \Delta V_j(k) &= z^T(k) A^T Q P Q A z(k) - z(k)^T Q P Q z(k)\\ \notag
        &= z^T(k) (A^T Q P Q A - Q P Q) z(k),
    \end{align*}
    \begin{equation} \notag
        \Delta V(k) = \sum_{j=1}^m \Delta V_j(k) = m z(k)^T (A^T Q P Q A - Q P Q) z(k),
    \end{equation}
    At consensus, $ A z(k) = z(k) $, so $ \Delta V(k) = 0 $. 
    
    For Lyapunov stability, it is required that $ \Delta V(k) < 0 $ when $ Q z(k) \neq 0 $ (i.e., outside consensus). This implies the LMI condition $ A^T Q P Q A - Q P Q < 0 $ in the disagreement subspace, which must hold for all $ z(k) $ such that $ Q z(k) \neq 0 $, where $P$ is defined in \eqref{eq:P_definition} with $ \delta > 0 $ ensuring positive definiteness.

%\section*{Acknowledgment}

%ChatGPT and Grammarly were used to assist in ensuring the manuscript is free of grammatical, punctuation, and tone-related errors.

\bibliographystyle{IEEEtran}
\bibliography{Mybib}

\end{document}